\documentclass{article} 
\usepackage[british]{babel}
\emergencystretch=\maxdimen
\usepackage{indentfirst}
\usepackage[utf8]{inputenc}
\usepackage[T1]{fontenc}

\usepackage{hyperref} 
\usepackage[top=1in, bottom=1in, left=1in, right=1in]{geometry}

\usepackage{lipsum}

\usepackage{graphicx}
\usepackage{epstopdf}
\usepackage{epsfig}
\usepackage{amsthm} 
\usepackage{amsmath} %
\usepackage{mathtools} 
\usepackage[adobe-utopia]{mathdesign}

\usepackage{multirow}

\newenvironment{myproof}[1][\proofname]{%
  \begin{proof}[#1]$ $\par\nobreak\ignorespaces   
}{%

  \end{proof}
}

\theoremstyle{definition}
\newtheorem{definition}{Definition}[section] %

\newtheorem{example}{Example}[section]

\newtheorem{theorem}{Theorem}[section]

\newcommand*{\titleGM}{\begingroup 
\begin{center}
\hbox{ 
\hspace*{0.05\textwidth} 
\parbox[b]{0.98\textwidth}
{ 
{\huge Four-divergence as a paravector operator.

Invariance of the wave equation under
\newline orthogonal paravector transformation. \\[0.5\baselineskip]}
\begin{center}
{\Large \textsc{Józef Radomański}} 
\\
e-mail: radomanski@wikakwa.pl
\end{center}
\begin{flushright}
I dedicate all my effort that I have spent on the search for the truth to Eliane Trevisan.
\end{flushright}

\vspace{0.05\textheight} 
}}

\end{center}
\begin{abstract}

\textit{The article presents four identities containing the differential operator $\partial^{\pm} =[\partial /\partial t, \pm \nabla ]$ also known as four-divergence. These equations are used to prove the invariance of wave equation under orthogonal paravector transformations. Moreover, the transformation of equations containing operators $\partial^{\pm}$ under the rotation of reference frame has been presented.}
\end{abstract}

Keywords: \textit{Four-divergence, paravector, wave equation, complex space-time, Clifford analysis}

\endgroup}

\begin{document}

\pagestyle{empty} 

\titleGM 

\pagestyle{plain} 


\section*{Introduction}
Although from the very beginning of the creation of the Special Theory of Relativity (STR) researchers tried to use various mathematical tools to describe the relativistic phenomenons, the mainstream of theoretical physics went in the direction of tensor calculus. It is a formalism that requires extensive mathematical knowledge and proficiency. Operations on vectors, even when they are more than 3-dimensional, are conceivable as opposed to the manipulation of rates, which is why many authors still try to describe the Theory of Relativity using more friendly tools. For some time now there have appeared a lot of articles on the STR, whose authors use the quaternion algebra, geometric algebra or paravectors, which proves that the tensor calculus has not been accepted by all as the best language for this branch of physics. 

J.L.Synge \cite{Synge} (1972) tried to use complex quaternions to present Maxwell's equations. At the same time, a lot of articles on relativistic physics were written by David Hestenes \cite{Hestenes} using Grassman Algebra. Following him, but using the paravectors, was \href{http://www1.uwindsor.ca/physics/dr-william-baylis}{William Baylis} who showed that paravectors and multivectors belong to different  representations of the same algebra \cite{Baylis}. Pedagogical experience of Professor William Baylis \cite{Baylis_1} shows that the STR taught in paravector formalism is much faster and more easily absorbed by the students than when he taught it traditionally. This means that this formalism is more intuitive, and therefore it suited better to describe the relativistic phenomena then tensor calculus.

This work is a continuation of my article ,,\href{http://arxiv.org/abs/1601.02965}{Algebra of paravectors}'' where paravectors are shown in such a way that everyone can imagine them as vectors. With the operation of summation paravectors create a unitary space over the field of complex numbers. Although it is a very important difference compared to the commonly applied definition of the scalar product (the product of paravector with itself is not a real positive number but is a complex number!),  the analogy is so clear that I decided to leave the name commonly used in linear algebra. From the point of view of algebra paravectors together with summation and multiplication create a ring with identity, due to which they have some characteristics of numbers.

In the current paper there are shown simple methods of transformation of the expressions containing differential operator under the linear transformation described by paravector. The identities containing the operator of spatio-temporal differentiation will be proved through which, in turn, invariance of wave equation will be shown under orthogonal transformations. Next, I will show how these equations are transformed under an Euclidean rotation. I will not interpret the results from the physics point of view nor will I analyse the physical sense of the domain or complex space. I treat the subject as pure mathematics to be used in the next works undertaking more practical problems. At the moment I can only say that the complex space-time, although similar to the Euclidean one, has a different structure, but is not Minkowski space either. Therefore, the results will come gradually.

Before reading this work, it is recommend to read the article ,,\href{http://arxiv.org/abs/1601.02965}{Algebra of paravectors}''\cite{Radomanski}.
        
\section{The spatio-temporal differential operator.}   

Since we do not know yet what restrictions should be put on the differentiation so it would not be in conflict with physics, we apply the most general assumption: The domain of the paravector function is a complex space-time. It means that both time and space are complex. Having to put some assumptions is evidenced by the fact that the differentiation of time should have other rules than the differentiation of space if only because time doesn't flinch. Determining when, why and what assumptions we need to put requires a thorough examination, but the subject is so vast that it will be discussed gradually  in subsequent publications.

In conclusion of the article ,,\href{http://arxiv.org/abs/1601.02965}{Algebra of paravectors}'' it was mentioned that some paravectors are additive, and others are not, in spite of the same construction. The additive ones are called traditionally: 4-vectors and are denoted with a capital doubled letter or featured as a column matrix in parentheses, for example:

\begin{equation}
\mathbb{X} :=
                            \begin{pmatrix}
                            \Delta t \\ 
                            \Delta \mathbf{x}%
                            \end{pmatrix}%
\end{equation}	

Non-additive paravectors which we can only multiply are denoted by a capital letter or presented as a column array in square brackets, for example:

\begin{equation}
               \Gamma :=
                \begin{bmatrix}
                \alpha \\ 
                \boldsymbol{\beta }
                \end{bmatrix}
                =
                \begin{bmatrix}
                a+id \\ 
                \textbf{b}+i\textbf{c}
                \end{bmatrix}
\end{equation} 
        
\begin{definition}
The \textbf{spatio-temporal differential operator} (or \textbf{4-divergence}) we call the paravector:
\[
\partial :=
\begin{bmatrix}
\frac{\partial }{\partial t} \\ 
\nabla
\end{bmatrix}
\]

We call the \textbf{4-gradient} operator reversed to the 4-divergence, or:
\[ \partial^{-} =
\begin{bmatrix}
\frac{\partial }{\partial t} \\ 
-\nabla
\end{bmatrix}
\]

Let $A(X)$ be an analytic paravector function defined on the set $C^{1+3}$. The spatio-temporal differential operator works on the function $A(X)$ as follows:

\begin{equation}
\partial A(X) =
\begin{bmatrix}
\frac{\partial }{\partial t} \\ 
\nabla
\end{bmatrix}
\begin{bmatrix}
\varphi (X) \\ 
\pmb{\Phi }(X)
\end{bmatrix}=
\begin{bmatrix}
\frac{\partial \varphi}{\partial t}+ \nabla \pmb{\Phi }\\ 
\frac{\partial \pmb{\Phi }}{\partial t}+\nabla \varphi +i\nabla \times \pmb{\Phi}
\end{bmatrix}
\end{equation}
\end{definition}

\begin{example}

In this notation the equations of electricity and magnetism in a vacuum has a form:
               \begin{equation}                    
                    \frac{1}{\epsilon_{0}}
                    \begin{pmatrix}
                    \rho \\ 
                    -\mathbf{j}/c
                    \end{pmatrix}%
                    =%
                    \begin{bmatrix}
                    \frac{\partial }{c\partial t} \\ 
                    \nabla%
                    \end{bmatrix}%
                    \begin{pmatrix}
                    0 \\ 
                    \mathbf{E}+ic\mathbf{B}
                    \end{pmatrix}
                    \qquad \text{i}\qquad 
                    \begin{pmatrix}
                    0 \\ 
                    \mathbf{E}+ic\mathbf{B}
                    \end{pmatrix}%
                    =%
                    \begin{bmatrix}
                    \frac{\partial }{c\partial t} \\ 
                    -\nabla%
                    \end{bmatrix}%
                    \begin{pmatrix}
                    \varphi \\ 
                    -c\mathbf{A}%
                    \end{pmatrix}%
               \end{equation}
  
After the operation as in equation (3) on the left side of (4) we obtain Maxwell's equations in a vacuum, and on the right side we get conditions so that the field around charges should describe a system of wave equations: 
 
\begin{equation}
                    (\frac{\partial^{2}}{\partial t^{2}}  - \nabla^{2} )
                    \begin{pmatrix}
                    \varphi \\ 
                    -c\mathbf{A}%
                    \end{pmatrix}=
                    \begin{bmatrix}
                    \frac{\partial }{c\partial t} \\ 
                    \nabla%
                    \end{bmatrix}%
                    \begin{bmatrix}
                    \frac{\partial }{c\partial t} \\ 
                    -\nabla%
                    \end{bmatrix}%
                    \begin{pmatrix}
                    \varphi \\ 
                    -c\mathbf{A}%
                    \end{pmatrix}%
                    =\frac{1}{\epsilon_{0}}
                    \begin{pmatrix}
                    \rho \\ 
                    -\mathbf{j}/c
                    \end{pmatrix}%
\end{equation}  
\end{example}

From the articles of Professor William Baylis and by the example above, one can see that the paravectors calculus is firmly rooted in physics and therefore it should be looked at it more closely from the mathematical point of view, so that this mathematics could be later applied in practice.

\section{Properties of the operator \texorpdfstring{$\partial$}{partial}}

\begin{theorem}
If paravector function $A(X)$ is analytic and additive, then
\begin{equation}
\partial [A_{1}(X)+A_{2}(X)]=\partial A_{1}(X)+\partial A_{2}(X)
\end{equation}
\end{theorem}
\begin{proof}

It is due to the fact that the derivative, gradient, divergence and curl keep the additivity of function.
\end{proof}

\begin{theorem}
Let the scalar function $\rho(x)$ and the paravector function $A(X)$ be analytic and defined on the set $C^{1+3}$, then
\begin{equation} \label{con1}
\partial [\rho(X)A(X)] = [\partial \rho(X)] A(X) + \rho(X) [\partial A(X)]
\end{equation}
\end{theorem}
\begin{myproof}
\[\partial [\rho(X)A(X)]=
\begin{bmatrix}
\frac{\partial }{\partial t} \\ 
\nabla
\end{bmatrix}
\begin{bmatrix}
\rho(X) \varphi(X) \\ 
\rho(X) \pmb{\Phi}(X)
\end{bmatrix}
=
\begin{bmatrix}
\frac{\partial (\rho \varphi)}{\partial t}+ \nabla(\rho \pmb{\Phi }) \\ 
\frac{\partial(\rho \pmb{\Phi })}{\partial t}+\nabla(\rho \varphi)+i\nabla \times(\rho \pmb{\Phi })
\end{bmatrix}=
\]
\[
=\begin{bmatrix}
\frac{\partial \rho}{\partial t}\varphi+\rho \frac{\partial  \varphi}{\partial t}
+ \pmb{\Phi }\nabla\rho +\rho\nabla\pmb{\Phi } \\ 
\frac{\partial\rho}{\partial t} \pmb{\Phi }
+\rho \frac{\partial\pmb{\Phi }}{\partial t}
+ \varphi\nabla\rho +\rho  \nabla\varphi
+i\rho (\nabla \times\pmb{\Phi }) + i(\nabla\rho)\times\pmb{\Phi}
\end{bmatrix}=
\begin{bmatrix}
\frac{\partial \rho}{\partial t}\varphi
+ \pmb{\Phi }\nabla\rho \\ 
\frac{\partial\rho}{\partial t} \pmb{\Phi }
+ \varphi\nabla\rho
+ i(\nabla\rho)\times\pmb{\Phi}
\end{bmatrix}+
\begin{bmatrix}
\rho \frac{\partial  \varphi}{\partial t}
+\rho\nabla\pmb{\Phi } \\ 
\rho \frac{\partial\pmb{\Phi }}{\partial t}
+\rho  \nabla\varphi
+i\rho (\nabla \times\pmb{\Phi })
\end{bmatrix}
=\]
\[=
\begin{bmatrix}
\frac{\partial \rho }{\partial t} \\ 
\nabla\rho
\end{bmatrix}
\begin{bmatrix}
\varphi \\ 
\pmb{\Phi}
\end{bmatrix} +
\rho \begin{bmatrix}
\frac{\partial }{\partial t} \\ 
\nabla
\end{bmatrix}
\begin{bmatrix}
\varphi \\ 
\pmb{\Phi}
\end{bmatrix}
\]
\end{myproof}

The operator $\partial^{-}$ holds similar property.

\textbf{Note:} Despite some similarities between the operator $\partial$ and the derivative of function of one variable, the properties of differential operator are not as extensive, for example:

\begin{itemize}
\item The formula \eqref{con1} is not true for the product of two paravector functions.
\item If on the left-hand side of the equation \eqref{con1} we reorder the scalar and paravector function, we can not do it on the right side.
\end{itemize}

Below, four identities will be shown by means of which one can show how to change the equation containing the spatio-temporal differential operator by the paravector transformation. The proofs are not complicated but are somewhat tedious, so we bring forth in detail only the first and third identities. I will take the reader a shortcut while proving the second one and will leave fourth one for self-proving.

\begin{theorem}\label{th:1}

Suppose that $A(X)$ is a paravector analytic function defined on the set $C^{1+3}$ and let the non-singular paravector $\Gamma$ determine the automorphism in the set $C^{1+3}$ so that $X^{\prime} = \Gamma X$, then the following identities are true:

\begin{enumerate}
\item 
$\partial A(X) =\partial ^{\prime }\Gamma A( \Gamma^{-1}X^{\prime })$\label{eq:th1.1}
\item
$\partial ^{-}A\left( X\right) =\Gamma ^{-}\partial ^{\prime -}A\left(\Gamma ^{-1}X^{\prime }\right)$\label{eq:th1.2}
\end{enumerate}
\end{theorem}

\begin{myproof}

Let's expand the equation $X^{\prime }=\Gamma X:$

\begin{flushleft}
\qquad $t^{\prime }=\alpha t+x\beta _{x}+y\beta _{y}+z\beta _{z}$

\qquad $x^{\prime }=t\beta _{x}+\alpha x-iy\beta _{z}+iz\beta _{y}$

\qquad $y^{\prime }=t\beta _{y}+ix\beta _{z}+\alpha y-iz\beta _{x}$

\qquad $z^{\prime }=t\beta _{z}-ix\beta_{y}+iy\beta _{x}+\alpha z$
\end{flushleft}
        
1. We transform the differential expression $\partial A(X)$

\begin{equation} \label{eq:1.1a}
\partial A(X)=\partial A(\Gamma ^{-1}\Gamma X)=\partial A(\Gamma
^{-1}X^{\prime })=
\begin{bmatrix}
\frac{\partial }{\partial t} \\ 
\nabla
\end{bmatrix}
\begin{bmatrix}
\varphi (\Gamma ^{-1}X^{\prime }) \\ 
\pmb{\Phi }(\Gamma ^{-1}X^{\prime })
\end{bmatrix}
=
\begin{bmatrix}
\frac{\partial \varphi ^{\prime }}{\partial t}+ \nabla \pmb{\Phi }^{\prime} \\ 
\frac{\partial \pmb{\Phi }^{\prime }}{\partial t}+\nabla \varphi ^{\prime}+i\nabla \times \pmb{\Phi }^{\prime }
\end{bmatrix}
\end{equation}

where the prime at the symbol of a function means that the argument is the phase $\Gamma^{-1}X^{\prime }$.

Using the formula for the derivative of the composite function we get:

\begin{equation}\label{eq:1.1b}
\frac{\partial \varphi ^{\prime }}{\partial t}=\frac{\partial \varphi
^{\prime }}{\partial t^{\prime }}\frac{\partial t^{\prime }}{\partial t}+\frac{\partial \varphi ^{\prime }}{\partial x^{\prime }}\frac{\partial x^{\prime }}{\partial t}+\frac{\partial \varphi ^{\prime }}{\partial y^{\prime }}\frac{\partial y^{\prime }}{\partial t}+\frac{\partial \varphi^{\prime }}{\partial z^{\prime }}\frac{\partial z^{\prime }}{\partial t}=
\frac{\partial \varphi ^{\prime }}{\partial t^{\prime }}\alpha +\frac{\partial \varphi ^{\prime }}{\partial x^{\prime }}\beta _{x}+\frac{\partial \varphi ^{\prime }}{\partial y^{\prime }}\beta _{y}+\frac{\partial \varphi^{\prime }}{\partial z^{\prime }}\beta _{z}=\frac{\partial \varphi ^{\prime }}{\partial t^{\prime }}\alpha +\pmb{\beta }\nabla ^{\prime }\varphi^{\prime }
\end{equation}

$\begin{array}{ccccc}
\nabla \Phi^{\prime } & =\frac{\partial \Phi _{x}^{\prime }}{\partial t^{\prime }}\frac{\partial t^{\prime }}{\partial x} & +\frac{\partial \Phi _{x}^{\prime }}{\partial x^{\prime }}\frac{\partial x^{\prime }}{\partial x} & +\frac{\partial \Phi
_{x}^{\prime }}{\partial y^{\prime }}\frac{\partial y^{\prime }}{\partial x} & +\frac{\partial \Phi _{x}^{\prime }}{\partial z^{\prime }}\frac{\partial z^{\prime }}{\partial x}+ \\ 
 &+\frac{\partial \Phi _{y}^{\prime }}{\partial t^{\prime }}\frac{\partial t^{\prime }}{\partial y} & +\frac{\partial \Phi _{y}^{\prime }}{\partial x^{\prime }}\frac{\partial x^{\prime }}{\partial y} & +\frac{\partial \Phi_{y}^{\prime }}{\partial y^{\prime }}\frac{\partial y^{\prime }}{\partial y} & +\frac{\partial \Phi _{y}^{\prime }}{\partial z^{\prime }}\frac{\partial z^{\prime }}{\partial y}+ \\ 
 &+\frac{\partial \Phi _{z}^{\prime }}{\partial t^{\prime }}\frac{\partial t^{\prime }}{\partial z} & +\frac{\partial \Phi _{z}^{\prime }}{\partial x^{\prime }}\frac{\partial x^{\prime }}{\partial z} & +\frac{\partial \Phi _{z}^{\prime }}{\partial y^{\prime }}\frac{\partial y^{\prime }}{\partial z} & +\frac{\partial \Phi _{z}^{\prime }}{\partial z^{\prime }}\frac{\partial z^{\prime }}{\partial z}= \\
 \\
 &=\frac{\partial \Phi _{x}^{\prime }}{\partial t^{\prime }}\beta _{x} & +\frac{\partial \Phi _{x}^{\prime }}{\partial x^{\prime }}\alpha & +i\beta _{z}\frac{\partial \Phi _{x}^{\prime }}{\partial y^{\prime }} & -i\beta _{y}\frac{\partial \Phi _{x}^{\prime }}{\partial z^{\prime }}+ \\ 
 & +\frac{\partial \Phi _{y}^{\prime }}{\partial t^{\prime }}\beta _{y} & -i\beta _{z}\frac{\partial \Phi _{y}^{\prime }}{\partial x^{\prime }} & +\frac{\partial \Phi _{y}^{\prime }}{\partial y^{\prime }}\alpha & +i\beta_{x}\frac{\partial \Phi _{y}^{\prime }}{\partial z^{\prime }}+ \\ 
 & +\frac{\partial \Phi _{z}^{\prime }}{\partial t^{\prime }}\beta _{z} & +i\beta _{y}\frac{\partial \Phi _{z}^{\prime }}{\partial x^{\prime }} & -i\beta _{x}\frac{\partial \Phi _{z}^{\prime }}{\partial y^{\prime }} & +\frac{\partial \Phi _{z}^{\prime }}{\partial z^{\prime }}\alpha=
\end{array}
$
\begin{equation} \label{eq:1.1c}
=\pmb{\beta }\frac{\partial \pmb{\Phi }^{\prime }}{\partial t^{\prime}}+\alpha \nabla ^{\prime } \pmb{\Phi }^{\prime }-i\pmb{\beta }\left(\nabla ^{\prime }\times \pmb{\Phi }^{\prime }\right) =\pmb{\beta }\frac{\partial \pmb{\Phi }^{\prime }}{\partial t^{\prime }}+\alpha \nabla^{\prime } \pmb{\Phi }^{\prime }+\nabla ^{\prime }i\left( \pmb{\beta }\times \pmb{\Phi }^{\prime }\right)
\end{equation}
\begin{flushleft}
$\nabla \varphi ^{\prime }=
\begin{bmatrix}
\frac{\partial \varphi ^{\prime }}{\partial t^{\prime }}\frac{\partial
t^{\prime }}{\partial x}+\frac{\partial \varphi ^{\prime }}{\partial
x^{\prime }}\frac{\partial x^{\prime }}{\partial x}+\frac{\partial \varphi
^{\prime }}{\partial y^{\prime }}\frac{\partial y^{\prime }}{\partial x}+
\frac{\partial \varphi ^{\prime }}{\partial z^{\prime }}\frac{\partial
z^{\prime }}{\partial x} \\ 
\frac{\partial \varphi ^{\prime }}{\partial t^{\prime }}\frac{\partial
t^{\prime }}{\partial y}+\frac{\partial \varphi ^{\prime }}{\partial
x^{\prime }}\frac{\partial x^{\prime }}{\partial y}+\frac{\partial \varphi
^{\prime }}{\partial y^{\prime }}\frac{\partial y^{\prime }}{\partial y}+
\frac{\partial \varphi ^{\prime }}{\partial z^{\prime }}\frac{\partial
z^{\prime }}{\partial y} \\ 
\frac{\partial \varphi ^{\prime }}{\partial t^{\prime }}\frac{\partial
t^{\prime }}{\partial z}+\frac{\partial \varphi ^{\prime }}{\partial
x^{\prime }}\frac{\partial x^{\prime }}{\partial z}+\frac{\partial \varphi
^{\prime }}{\partial y^{\prime }}\frac{\partial y^{\prime }}{\partial z}+
\frac{\partial \varphi ^{\prime }}{\partial z^{\prime }}\frac{\partial
z^{\prime }}{\partial z}
\end{bmatrix}
=
\begin{bmatrix}
\frac{\partial \varphi ^{\prime }}{\partial t^{\prime }}\beta _{x}+\frac{%
\partial \varphi ^{\prime }}{\partial x^{\prime }}\alpha +i\beta _{z}\frac{%
\partial \varphi ^{\prime }}{\partial y^{\prime }}-i\beta _{y}\frac{\partial
\varphi ^{\prime }}{\partial z^{\prime }} \\ 
\frac{\partial \varphi ^{\prime }}{\partial t^{\prime }}\beta _{y}-i\beta
_{z}\frac{\partial \varphi ^{\prime }}{\partial x^{\prime }}+\frac{\partial
\varphi ^{\prime }}{\partial y^{\prime }}\alpha +i\beta _{x}\frac{\partial
\varphi ^{\prime }}{\partial z^{\prime }} \\ 
\frac{\partial \varphi ^{\prime }}{\partial t^{\prime }}\beta _{z}+i\beta
_{y}\frac{\partial \varphi ^{\prime }}{\partial x^{\prime }}-i\beta _{x}%
\frac{\partial \varphi ^{\prime }}{\partial y^{\prime }}+\frac{\partial \varphi ^{\prime }}{\partial
z^{\prime }}\alpha%
\end{bmatrix}%
=$
\end{flushleft}
\begin{equation}\label{eq:1.1d}
=\pmb{\beta }\frac{\partial \varphi ^{\prime }}{\partial t^{\prime }}+\alpha \nabla ^{\prime }\varphi ^{\prime }+i\left( \nabla ^{\prime }\varphi^{\prime }\right) \times \pmb{\beta }=\pmb{\beta }\frac{\partial
\varphi ^{\prime }}{\partial t^{\prime }}+\alpha \nabla ^{\prime }\varphi
^{\prime }+i\nabla ^{\prime }\times \left( \pmb{\beta }\varphi ^{\prime
}\right)
\end{equation}
\[
\frac{\partial \pmb{\Phi} ^{\prime }}{\partial t}+i\nabla \times \pmb{\Phi} %
^{\prime }=\frac{\partial \pmb{\Phi} ^{\prime }}{\partial t}+i%
\begin{bmatrix}
\frac{\partial \Phi _{z}^{\prime }}{\partial y}-\frac{\partial \Phi
_{y}^{\prime }}{\partial z} \\ 
\frac{\partial \Phi _{x}^{\prime }}{\partial z}-\frac{\partial \Phi
_{z}^{\prime }}{\partial x} \\ 
\frac{\partial \Phi _{y}^{\prime }}{\partial x}-\frac{\partial \Phi
_{x}^{\prime }}{\partial y}%
\end{bmatrix}%
=\qquad \qquad \qquad \qquad \qquad \qquad \qquad \qquad \qquad \qquad\qquad \qquad \qquad
\]

$=
\begin{bmatrix}
\frac{\partial \Phi _{x}^{\prime }}{\partial t^{\prime }}\frac{\partial
t^{\prime }}{\partial t}+\frac{\partial \Phi _{x}^{\prime }}{\partial
x^{\prime }}\frac{\partial x^{\prime }}{\partial t}+\frac{\partial \Phi
_{x}^{\prime }}{\partial y^{\prime }}\frac{\partial y^{\prime }}{\partial t}+%
\frac{\partial \Phi _{x}^{\prime }}{\partial z^{\prime }}\frac{\partial
z^{\prime }}{\partial t} \\ 
\frac{\partial \Phi _{y}^{\prime }}{\partial t^{\prime }}\frac{\partial
t^{\prime }}{\partial t}+\frac{\partial \Phi _{y}^{\prime }}{\partial
x^{\prime }}\frac{\partial x^{\prime }}{\partial t}+\frac{\partial \Phi
_{y}^{\prime }}{\partial y^{\prime }}\frac{\partial y^{\prime }}{\partial t}+%
\frac{\partial \Phi _{y}^{\prime }}{\partial z^{\prime }}\frac{\partial
z^{\prime }}{\partial t} \\ 
\frac{\partial \Phi _{z}^{\prime }}{\partial t^{\prime }}\frac{\partial
t^{\prime }}{\partial t}+\frac{\partial \Phi _{z}^{\prime }}{\partial
x^{\prime }}\frac{\partial x^{\prime }}{\partial t}+\frac{\partial \Phi
_{z}^{\prime }}{\partial y^{\prime }}\frac{\partial y^{\prime }}{\partial t}+%
\frac{\partial \Phi _{z}^{\prime }}{\partial z^{\prime }}\frac{\partial
z^{\prime }}{\partial t}%
\end{bmatrix}+$

$+i\begin{bmatrix}
\frac{\partial \Phi _{z}^{\prime }}{\partial t^{\prime }}\frac{\partial
t^{\prime }}{\partial y}+\frac{\partial \Phi _{z}^{\prime }}{\partial
x^{\prime }}\frac{\partial x^{\prime }}{\partial y}+\frac{\partial \Phi
_{z}^{\prime }}{\partial y^{\prime }}\frac{\partial y^{\prime }}{\partial y}+%
\frac{\partial \Phi _{z}^{\prime }}{\partial z^{\prime }}\frac{\partial
z^{\prime }}{\partial y}-\frac{\partial \Phi _{y}^{\prime }}{\partial
t^{\prime }}\frac{\partial t^{\prime }}{\partial z}-\frac{\partial \Phi
_{y}^{\prime }}{\partial x^{\prime }}\frac{\partial x^{\prime }}{\partial z}-%
\frac{\partial \Phi _{y}^{\prime }}{\partial y^{\prime }}\frac{\partial
y^{\prime }}{\partial z}-\frac{\partial \Phi _{y}^{\prime }}{\partial
z^{\prime }}\frac{\partial z^{\prime }}{\partial z} \\ 
\frac{\partial \Phi _{x}^{\prime }}{\partial t^{\prime }}\frac{\partial
t^{\prime }}{\partial z}+\frac{\partial \Phi _{x}^{\prime }}{\partial
x^{\prime }}\frac{\partial x^{\prime }}{\partial z}+\frac{\partial \Phi
_{x}^{\prime }}{\partial y^{\prime }}\frac{\partial y^{\prime }}{\partial z}+%
\frac{\partial \Phi _{x}^{\prime }}{\partial z^{\prime }}\frac{\partial
z^{\prime }}{\partial z}-\frac{\partial \Phi _{z}^{\prime }}{\partial
t^{\prime }}\frac{\partial t^{\prime }}{\partial x}-\frac{\partial \Phi
_{z}^{\prime }}{\partial x^{\prime }}\frac{\partial x^{\prime }}{\partial x}-%
\frac{\partial \Phi _{z}^{\prime }}{\partial y^{\prime }}\frac{\partial
y^{\prime }}{\partial x}-\frac{\partial \Phi _{z}^{\prime }}{\partial
z^{\prime }}\frac{\partial z^{\prime }}{\partial x} \\ 
\frac{\partial \Phi _{y}^{\prime }}{\partial t^{\prime }}\frac{\partial
t^{\prime }}{\partial x}+\frac{\partial \Phi _{y}^{\prime }}{\partial
x^{\prime }}\frac{\partial x^{\prime }}{\partial x}+\frac{\partial \Phi
_{y}^{\prime }}{\partial y^{\prime }}\frac{\partial y^{\prime }}{\partial x}+%
\frac{\partial \Phi _{y}^{\prime }}{\partial z^{\prime }}\frac{\partial
z^{\prime }}{\partial x}-\frac{\partial \Phi _{x}^{\prime }}{\partial
t^{\prime }}\frac{\partial t^{\prime }}{\partial y}-\frac{\partial \Phi
_{x}^{\prime }}{\partial x^{\prime }}\frac{\partial x^{\prime }}{\partial y}-%
\frac{\partial \Phi _{x}^{\prime }}{\partial y^{\prime }}\frac{\partial
y^{\prime }}{\partial y}-\frac{\partial \Phi _{x}^{\prime }}{\partial
z^{\prime }}\frac{\partial z^{\prime }}{\partial y}%
\end{bmatrix}%
=$

\noindent $=\begin{bmatrix}
\frac{\partial \Phi _{x}^{\prime }}{\partial t^{\prime }}\alpha +\frac{%
\partial \Phi _{x}^{\prime }}{\partial x^{\prime }}\beta _{x}+\frac{\partial
\Phi _{x}^{\prime }}{\partial y^{\prime }}\beta _{y}+\frac{\partial \Phi
_{x}^{\prime }}{\partial z^{\prime }}\beta _{z} \\ 
\frac{\partial \Phi _{y}^{\prime }}{\partial t^{\prime }}\alpha +\frac{%
\partial \Phi _{y}^{\prime }}{\partial x^{\prime }}\beta _{x}+\frac{\partial
\Phi _{y}^{\prime }}{\partial y^{\prime }}\beta _{y}+\frac{\partial \Phi
_{y}^{\prime }}{\partial z^{\prime }}\beta _{z} \\ 
\frac{\partial \Phi _{z}^{\prime }}{\partial t^{\prime }}\alpha +\frac{%
\partial \Phi _{z}^{\prime }}{\partial x^{\prime }}\beta _{x}+\frac{\partial
\Phi _{z}^{\prime }}{\partial y^{\prime }}\beta _{y}+\frac{\partial \Phi
_{z}^{\prime }}{\partial z^{\prime }}\beta _{z}%
\end{bmatrix}
+i
\begin{bmatrix}
\frac{\partial \Phi _{z}^{\prime }}{\partial t^{\prime }}\beta _{y}-i\beta
_{z}\frac{\partial \Phi _{z}^{\prime }}{\partial x^{\prime }}+\frac{\partial
\Phi _{z}^{\prime }}{\partial y^{\prime }}\alpha +i\beta _{x}\frac{\partial
\Phi _{z}^{\prime }}{\partial z^{\prime }}-\frac{\partial \Phi _{y}^{\prime }%
}{\partial t^{\prime }}\beta _{z}-i\beta _{y}\frac{\partial \Phi
_{y}^{\prime }}{\partial x^{\prime }}+i\beta _{x}\frac{\partial \Phi
_{y}^{\prime }}{\partial y^{\prime }}-\frac{\partial \Phi _{y}^{\prime }}{%
\partial z^{\prime }}\alpha  \\ 
\frac{\partial \Phi _{x}^{\prime }}{\partial t^{\prime }}\beta _{z}+i\beta
_{y}\frac{\partial \Phi _{x}^{\prime }}{\partial x^{\prime }}-i\beta _{x}%
\frac{\partial \Phi _{x}^{\prime }}{\partial y^{\prime }}+\frac{\partial
\Phi _{x}^{\prime }}{\partial z^{\prime }}\alpha -\frac{\partial \Phi
_{z}^{\prime }}{\partial t^{\prime }}\beta _{x}-\frac{\partial \Phi
_{z}^{\prime }}{\partial x^{\prime }}\alpha -i\beta _{z}\frac{\partial \Phi
_{z}^{\prime }}{\partial y^{\prime }}+i\beta _{y}\frac{\partial \Phi
_{z}^{\prime }}{\partial z^{\prime }} \\ 
\frac{\partial \Phi _{y}^{\prime }}{\partial t^{\prime }}\beta _{x}+\frac{%
\partial \Phi _{y}^{\prime }}{\partial x^{\prime }}\alpha +i\beta _{z}\frac{%
\partial \Phi _{y}^{\prime }}{\partial y^{\prime }}-i\beta _{y}\frac{%
\partial \Phi _{y}^{\prime }}{\partial z^{\prime }}-\frac{\partial \Phi
_{x}^{\prime }}{\partial t^{\prime }}\beta _{y}+i\beta _{z}\frac{\partial
\Phi _{x}^{\prime }}{\partial x^{\prime }}-\frac{\partial \Phi _{x}^{\prime }%
}{\partial y^{\prime }}\alpha -i\beta _{x}\frac{\partial \Phi _{x}^{\prime }%
}{\partial z^{\prime }}%
\end{bmatrix}%
=$

\begin{equation}\label{eq:1.1e}
=\alpha \frac{\partial \pmb{\Phi} ^{\prime }}{\partial t^{\prime }}+i\alpha
\nabla ^{\prime }\times \pmb{\Phi }^{\prime }-i\frac{\partial \pmb{
\Phi}^{\prime }}{\partial t^{\prime }}\times \pmb{\beta }+(\nabla
^{\prime }\pmb{\beta \Phi }^{\prime })+\nabla ^{\prime }\times (\pmb{
\Phi}^{\prime }\times \pmb{\beta })
\end{equation}

Substituting partial results \eqref{eq:1.1b} - \eqref{eq:1.1e} into the equation \eqref{eq:1.1a} we receive

\[
\begin{bmatrix}
\frac{\partial }{\partial t} \\ 
\nabla
\end{bmatrix}
\begin{bmatrix}
\varphi (X) \\ 
\pmb{\Phi }(X)
\end{bmatrix}=
\begin{bmatrix}
\frac{\partial \left( \alpha \varphi ^{\prime }+\pmb{\beta \pmb{\Phi} }^{\prime
}\right) }{\partial t^{\prime }}+\nabla ^{\prime }\left( \pmb{\beta }%
\varphi ^{\prime }+\alpha \pmb{\Phi }^{\prime }+i\pmb{\beta }\times 
\pmb{\Phi }^{\prime }\right) \\ 
\frac{\partial }{\partial t^{\prime }}\left( \pmb{\beta }\varphi ^{\prime
}+\alpha \pmb{\Phi }^{\prime }+i\pmb{\beta }\times \pmb{\Phi }%
^{\prime }\right) +\nabla ^{\prime }\left( \alpha \varphi ^{\prime }+\pmb{
\beta \Phi }^{\prime }\right) +i\nabla ^{\prime }\times \left( \pmb{\beta 
}\varphi ^{\prime }+\alpha \pmb{\Phi }^{\prime }+i\pmb{\beta }\times 
\pmb{\Phi }^{\prime }\right)%
\end{bmatrix}
=
\begin{bmatrix}
\frac{\partial }{\partial t^{\prime}} \\ 
\nabla^{\prime}%
\end{bmatrix}%
(\begin{bmatrix}
\alpha \\ 
\pmb{\beta }%
\end{bmatrix}%
\begin{bmatrix}
\varphi ^{\prime } \\ 
\pmb{\Phi }^{\prime }%
\end{bmatrix}),
\]

which completes the proof of 1st identity.

2. To prove the truth of the identity $\partial ^{-}A\left( X\right) =\Gamma ^{-}\partial ^{\prime -}A\left(\Gamma ^{-1}X^{\prime }\right)$, we must use the formulas \eqref {eq:1.1b} - \eqref {eq:1.1d}, and instead of \eqref{eq:1.1e} we must prove that:

\[
\frac{\partial \Phi ^{\prime }}{\partial t}-i\nabla \times \pmb{\Phi }^{\prime }=\alpha \frac{\partial \pmb{\Phi} ^{\prime }}{\partial t^{\prime }}-i\alpha \nabla ^{\prime }\times \pmb{\Phi }^{\prime }-i \pmb{\beta}
\times \frac{\partial \pmb{\Phi }^{\prime }}{\partial t^{\prime }}
+\pmb{\beta} (\nabla ^{\prime }\pmb{\Phi }^{\prime })+\pmb{\beta}
\times (\nabla ^{\prime }\times \pmb{\Phi }^{\prime })
\]
\end{myproof}

\begin{theorem}\label{th:2}

Suppose that $A(X)$ is a paravector analytic function defined on the set $C^{1+3}$ and let the non-singular paravector $\Gamma$ determine the automorphism in the set $C^{1+3}$ so that $X^{\prime} =X \Gamma $, then the following identities are true:

\begin{enumerate}
\item 
$\partial A\left( X\right) =\Gamma \partial ^{\prime }A\left( X^{\prime}\Gamma ^{-1}\right)$\label{eq:th2.1}
\item
$\partial ^{-}A\left( X\right) =\partial ^{\prime -}\Gamma ^{-}A\left(X^{\prime }\Gamma ^{-1}\right)$\label{eq:th2.2}
\end{enumerate}
\end{theorem}

\begin{myproof}

Let's expand the equation $X^{\prime }= X\Gamma:$

\begin{flushleft}
\qquad $t^{\prime }=\alpha t+x\beta _{x}+y\beta _{y}+z\beta _{z}$

\qquad $x^{\prime }=t\beta _{x}+\alpha x+iy\beta _{z}-iz\beta _{y}$

\qquad $y^{\prime }=t\beta _{y}-ix\beta _{z}+\alpha y+iz\beta _{x}$

\qquad $z^{\prime }=t\beta _{z}+ix\beta y-iy\beta _{x}+\alpha z$
\end{flushleft}

1. We transform the differential expression $\partial A(X)$

\begin{equation} \label{eq:2.1a}
\partial A(X)=\partial A(X\Gamma \Gamma^{-1} )=\partial A(X^{\prime }\Gamma^{-1})=
\begin{bmatrix}
\frac{\partial }{\partial t} \\ 
\nabla
\end{bmatrix}
\begin{bmatrix}
\varphi (X^{\prime }\Gamma^{-1}) \\ 
\pmb{\Phi }(X^{\prime }\Gamma^{-1})
\end{bmatrix}
=
\begin{bmatrix}
\frac{\partial \varphi ^{\prime }}{\partial t}+ \nabla \pmb{\Phi }^{\prime} \\ 
\frac{\partial \pmb{\Phi }^{\prime }}{\partial t}+\nabla \varphi ^{\prime}+i\nabla \times \pmb{\Phi }^{\prime }
\end{bmatrix}
\end{equation}

where the prime at the symbol of a function means that the argument is the phase $X^{\prime }\Gamma^{-1}$

Using the formula for the derivative of the composite function we get:

\begin{equation}\label{eq:2.1b}
\frac{\partial \varphi ^{\prime }}{\partial t}=\frac{\partial \varphi
^{\prime }}{\partial t^{\prime }}\frac{\partial t^{\prime }}{\partial t}+\frac{\partial \varphi ^{\prime }}{\partial x^{\prime }}\frac{\partial x^{\prime }}{\partial t}+\frac{\partial \varphi ^{\prime }}{\partial y^{\prime }}\frac{\partial y^{\prime }}{\partial t}+\frac{\partial \varphi^{\prime }}{\partial z^{\prime }}\frac{\partial z^{\prime }}{\partial t}=
\frac{\partial \varphi ^{\prime }}{\partial t^{\prime }}\alpha +\frac{\partial \varphi ^{\prime }}{\partial x^{\prime }}\beta _{x}+\frac{\partial \varphi ^{\prime }}{\partial y^{\prime }}\beta _{y}+\frac{\partial \varphi^{\prime }}{\partial z^{\prime }}\beta _{z}=\frac{\partial \varphi ^{\prime }}{\partial t^{\prime }}\alpha +\pmb{\beta }\nabla ^{\prime }\varphi^{\prime }
\end{equation}

$\begin{array}{ccccc}
\nabla \Phi^{\prime } & =\frac{\partial \Phi _{x}^{\prime }}{\partial t^{\prime }}\frac{\partial t^{\prime }}{\partial x} & +\frac{\partial \Phi _{x}^{\prime }}{\partial x^{\prime }}\frac{\partial x^{\prime }}{\partial x} & +\frac{\partial \Phi
_{x}^{\prime }}{\partial y^{\prime }}\frac{\partial y^{\prime }}{\partial x} & +\frac{\partial \Phi _{x}^{\prime }}{\partial z^{\prime }}\frac{\partial z^{\prime }}{\partial x}+ \\ 
 &+\frac{\partial \Phi _{y}^{\prime }}{\partial t^{\prime }}\frac{\partial t^{\prime }}{\partial y} & +\frac{\partial \Phi _{y}^{\prime }}{\partial x^{\prime }}\frac{\partial x^{\prime }}{\partial y} & +\frac{\partial \Phi_{y}^{\prime }}{\partial y^{\prime }}\frac{\partial y^{\prime }}{\partial y} & +\frac{\partial \Phi _{y}^{\prime }}{\partial z^{\prime }}\frac{\partial z^{\prime }}{\partial y}+ \\ 
 &+\frac{\partial \Phi _{z}^{\prime }}{\partial t^{\prime }}\frac{\partial t^{\prime }}{\partial z} & +\frac{\partial \Phi _{z}^{\prime }}{\partial x^{\prime }}\frac{\partial x^{\prime }}{\partial z} & +\frac{\partial \Phi _{z}^{\prime }}{\partial y^{\prime }}\frac{\partial y^{\prime }}{\partial z} & +\frac{\partial \Phi _{z}^{\prime }}{\partial z^{\prime }}\frac{\partial z^{\prime }}{\partial z}= \\
 \\
 &=\frac{\partial \Phi _{x}^{\prime }}{\partial t^{\prime }}\beta _{x} & +\frac{\partial \Phi _{x}^{\prime }}{\partial x^{\prime }}\alpha & -i\beta _{z}\frac{\partial \Phi _{x}^{\prime }}{\partial y^{\prime }} & +i\beta _{y}\frac{\partial \Phi _{x}^{\prime }}{\partial z^{\prime }}+ \\ 
 & +\frac{\partial \Phi _{y}^{\prime }}{\partial t^{\prime }}\beta _{y} & +i\beta _{z}\frac{\partial \Phi _{y}^{\prime }}{\partial x^{\prime }} & +\frac{\partial \Phi _{y}^{\prime }}{\partial y^{\prime }}\alpha & -i\beta_{x}\frac{\partial \Phi _{y}^{\prime }}{\partial z^{\prime }}+ \\ 
 & +\frac{\partial \Phi _{z}^{\prime }}{\partial t^{\prime }}\beta _{z} & -i\beta _{y}\frac{\partial \Phi _{z}^{\prime }}{\partial x^{\prime }} & +i\beta _{x}\frac{\partial \Phi _{z}^{\prime }}{\partial y^{\prime }} & +\frac{\partial \Phi _{z}^{\prime }}{\partial z^{\prime }}\alpha=
\end{array}
$
\begin{equation} \label{eq:2.1c}
=\pmb{\beta }\frac{\partial \pmb{\Phi }^{\prime }}{\partial t^{\prime}}+\alpha \nabla ^{\prime } \pmb{\Phi }^{\prime }+i\pmb{\beta }\left(\nabla ^{\prime }\times \pmb{\Phi }^{\prime }\right)
\end{equation}
\begin{flushleft}
$\nabla \varphi ^{\prime }=
\begin{bmatrix}
\frac{\partial \varphi ^{\prime }}{\partial t^{\prime }}\frac{\partial
t^{\prime }}{\partial x}+\frac{\partial \varphi ^{\prime }}{\partial
x^{\prime }}\frac{\partial x^{\prime }}{\partial x}+\frac{\partial \varphi
^{\prime }}{\partial y^{\prime }}\frac{\partial y^{\prime }}{\partial x}+
\frac{\partial \varphi ^{\prime }}{\partial z^{\prime }}\frac{\partial
z^{\prime }}{\partial x} \\ 
\frac{\partial \varphi ^{\prime }}{\partial t^{\prime }}\frac{\partial
t^{\prime }}{\partial y}+\frac{\partial \varphi ^{\prime }}{\partial
x^{\prime }}\frac{\partial x^{\prime }}{\partial y}+\frac{\partial \varphi
^{\prime }}{\partial y^{\prime }}\frac{\partial y^{\prime }}{\partial y}+
\frac{\partial \varphi ^{\prime }}{\partial z^{\prime }}\frac{\partial
z^{\prime }}{\partial y} \\ 
\frac{\partial \varphi ^{\prime }}{\partial t^{\prime }}\frac{\partial
t^{\prime }}{\partial z}+\frac{\partial \varphi ^{\prime }}{\partial
x^{\prime }}\frac{\partial x^{\prime }}{\partial z}+\frac{\partial \varphi
^{\prime }}{\partial y^{\prime }}\frac{\partial y^{\prime }}{\partial z}+
\frac{\partial \varphi ^{\prime }}{\partial z^{\prime }}\frac{\partial
z^{\prime }}{\partial z}
\end{bmatrix}
=
\begin{bmatrix}
\frac{\partial \varphi ^{\prime }}{\partial t^{\prime }}\beta _{x}+\frac{\partial \varphi ^{\prime }}{\partial x^{\prime }}\alpha -i\beta _{z}\frac{\partial \varphi ^{\prime }}{\partial y^{\prime }}+i\beta _{y}\frac{\partial\varphi ^{\prime }}{\partial z^{\prime }} \\ 
\frac{\partial \varphi ^{\prime }}{\partial t^{\prime }}\beta _{y}+i\beta_{z}\frac{\partial \varphi ^{\prime }}{\partial x^{\prime }}+\frac{\partial\varphi ^{\prime }}{\partial y^{\prime }}\alpha -i\beta _{x}\frac{\partial\varphi ^{\prime }}{\partial z^{\prime }} \\ 
\frac{\partial \varphi ^{\prime }}{\partial t^{\prime }}\beta _{z}-i\beta_{y}\frac{\partial \varphi ^{\prime }}{\partial x^{\prime }}+i\beta _{x}\frac{\partial \varphi ^{\prime }}{\partial y^{\prime }}+\frac{\partial \varphi ^{\prime }}{\partial z^{\prime }}\alpha
\end{bmatrix}
=$
\end{flushleft}
\begin{equation}\label{eq:2.1d}
=\pmb{\beta }\frac{\partial \varphi ^{\prime }}{\partial t^{\prime }}+\alpha \nabla ^{\prime }\varphi ^{\prime }+i\pmb{\beta } \times \left( \nabla ^{\prime }\varphi^{\prime }\right)
\end{equation}
\[
\frac{\partial \pmb{\Phi} ^{\prime }}{\partial t}+i\nabla \times \pmb{\Phi} %
^{\prime }=\frac{\partial \pmb{\Phi} ^{\prime }}{\partial t}+i%
\begin{bmatrix}
\frac{\partial \Phi _{z}^{\prime }}{\partial y}-\frac{\partial \Phi
_{y}^{\prime }}{\partial z} \\ 
\frac{\partial \Phi _{x}^{\prime }}{\partial z}-\frac{\partial \Phi
_{z}^{\prime }}{\partial x} \\ 
\frac{\partial \Phi _{y}^{\prime }}{\partial x}-\frac{\partial \Phi
_{x}^{\prime }}{\partial y}%
\end{bmatrix}%
=\qquad \qquad \qquad \qquad \qquad \qquad \qquad \qquad \qquad \qquad\qquad \qquad \qquad
\]

$=
\begin{bmatrix}
\frac{\partial \Phi _{x}^{\prime }}{\partial t^{\prime }}\frac{\partial
t^{\prime }}{\partial t}+\frac{\partial \Phi _{x}^{\prime }}{\partial
x^{\prime }}\frac{\partial x^{\prime }}{\partial t}+\frac{\partial \Phi
_{x}^{\prime }}{\partial y^{\prime }}\frac{\partial y^{\prime }}{\partial t}+%
\frac{\partial \Phi _{x}^{\prime }}{\partial z^{\prime }}\frac{\partial
z^{\prime }}{\partial t} \\ 
\frac{\partial \Phi _{y}^{\prime }}{\partial t^{\prime }}\frac{\partial
t^{\prime }}{\partial t}+\frac{\partial \Phi _{y}^{\prime }}{\partial
x^{\prime }}\frac{\partial x^{\prime }}{\partial t}+\frac{\partial \Phi
_{y}^{\prime }}{\partial y^{\prime }}\frac{\partial y^{\prime }}{\partial t}+%
\frac{\partial \Phi _{y}^{\prime }}{\partial z^{\prime }}\frac{\partial
z^{\prime }}{\partial t} \\ 
\frac{\partial \Phi _{z}^{\prime }}{\partial t^{\prime }}\frac{\partial
t^{\prime }}{\partial t}+\frac{\partial \Phi _{z}^{\prime }}{\partial
x^{\prime }}\frac{\partial x^{\prime }}{\partial t}+\frac{\partial \Phi
_{z}^{\prime }}{\partial y^{\prime }}\frac{\partial y^{\prime }}{\partial t}+%
\frac{\partial \Phi _{z}^{\prime }}{\partial z^{\prime }}\frac{\partial
z^{\prime }}{\partial t}%
\end{bmatrix}+$

$+i\begin{bmatrix}
\frac{\partial \Phi _{z}^{\prime }}{\partial t^{\prime }}\frac{\partial
t^{\prime }}{\partial y}+\frac{\partial \Phi _{z}^{\prime }}{\partial
x^{\prime }}\frac{\partial x^{\prime }}{\partial y}+\frac{\partial \Phi
_{z}^{\prime }}{\partial y^{\prime }}\frac{\partial y^{\prime }}{\partial y}+%
\frac{\partial \Phi _{z}^{\prime }}{\partial z^{\prime }}\frac{\partial
z^{\prime }}{\partial y}-\frac{\partial \Phi _{y}^{\prime }}{\partial
t^{\prime }}\frac{\partial t^{\prime }}{\partial z}-\frac{\partial \Phi
_{y}^{\prime }}{\partial x^{\prime }}\frac{\partial x^{\prime }}{\partial z}-%
\frac{\partial \Phi _{y}^{\prime }}{\partial y^{\prime }}\frac{\partial
y^{\prime }}{\partial z}-\frac{\partial \Phi _{y}^{\prime }}{\partial
z^{\prime }}\frac{\partial z^{\prime }}{\partial z} \\ 
\frac{\partial \Phi _{x}^{\prime }}{\partial t^{\prime }}\frac{\partial
t^{\prime }}{\partial z}+\frac{\partial \Phi _{x}^{\prime }}{\partial
x^{\prime }}\frac{\partial x^{\prime }}{\partial z}+\frac{\partial \Phi
_{x}^{\prime }}{\partial y^{\prime }}\frac{\partial y^{\prime }}{\partial z}+%
\frac{\partial \Phi _{x}^{\prime }}{\partial z^{\prime }}\frac{\partial
z^{\prime }}{\partial z}-\frac{\partial \Phi _{z}^{\prime }}{\partial
t^{\prime }}\frac{\partial t^{\prime }}{\partial x}-\frac{\partial \Phi
_{z}^{\prime }}{\partial x^{\prime }}\frac{\partial x^{\prime }}{\partial x}-%
\frac{\partial \Phi _{z}^{\prime }}{\partial y^{\prime }}\frac{\partial
y^{\prime }}{\partial x}-\frac{\partial \Phi _{z}^{\prime }}{\partial
z^{\prime }}\frac{\partial z^{\prime }}{\partial x} \\ 
\frac{\partial \Phi _{y}^{\prime }}{\partial t^{\prime }}\frac{\partial
t^{\prime }}{\partial x}+\frac{\partial \Phi _{y}^{\prime }}{\partial
x^{\prime }}\frac{\partial x^{\prime }}{\partial x}+\frac{\partial \Phi
_{y}^{\prime }}{\partial y^{\prime }}\frac{\partial y^{\prime }}{\partial x}+%
\frac{\partial \Phi _{y}^{\prime }}{\partial z^{\prime }}\frac{\partial
z^{\prime }}{\partial x}-\frac{\partial \Phi _{x}^{\prime }}{\partial
t^{\prime }}\frac{\partial t^{\prime }}{\partial y}-\frac{\partial \Phi
_{x}^{\prime }}{\partial x^{\prime }}\frac{\partial x^{\prime }}{\partial y}-%
\frac{\partial \Phi _{x}^{\prime }}{\partial y^{\prime }}\frac{\partial
y^{\prime }}{\partial y}-\frac{\partial \Phi _{x}^{\prime }}{\partial
z^{\prime }}\frac{\partial z^{\prime }}{\partial y}%
\end{bmatrix}%
=$

\noindent $=\begin{bmatrix}
\frac{\partial \Phi _{x}^{\prime }}{\partial t^{\prime }}\alpha +\frac{%
\partial \Phi _{x}^{\prime }}{\partial x^{\prime }}\beta _{x}+\frac{\partial
\Phi _{x}^{\prime }}{\partial y^{\prime }}\beta _{y}+\frac{\partial \Phi
_{x}^{\prime }}{\partial z^{\prime }}\beta _{z} \\ 
\frac{\partial \Phi _{y}^{\prime }}{\partial t^{\prime }}\alpha +\frac{%
\partial \Phi _{y}^{\prime }}{\partial x^{\prime }}\beta _{x}+\frac{\partial
\Phi _{y}^{\prime }}{\partial y^{\prime }}\beta _{y}+\frac{\partial \Phi
_{y}^{\prime }}{\partial z^{\prime }}\beta _{z} \\ 
\frac{\partial \Phi _{z}^{\prime }}{\partial t^{\prime }}\alpha +\frac{%
\partial \Phi _{z}^{\prime }}{\partial x^{\prime }}\beta _{x}+\frac{\partial
\Phi _{z}^{\prime }}{\partial y^{\prime }}\beta _{y}+\frac{\partial \Phi
_{z}^{\prime }}{\partial z^{\prime }}\beta _{z}%
\end{bmatrix}
+i
\begin{bmatrix}
\frac{\partial \Phi _{z}^{\prime }}{\partial t^{\prime }}\beta _{y}+i\beta_{z}\frac{\partial \Phi _{z}^{\prime }}{\partial x^{\prime }}+\frac{\partial \Phi _{z}^{\prime }}{\partial y^{\prime }}\alpha -i\beta _{x}\frac{\partial \Phi _{z}^{\prime }}{\partial z^{\prime }}-\frac{\partial \Phi _{y}^{\prime }}{\partial t^{\prime }}\beta _{z}+i\beta _{y}\frac{\partial \Phi
_{y}^{\prime }}{\partial x^{\prime }}-i\beta _{x}\frac{\partial \Phi_{y}^{\prime }}{\partial y^{\prime }}-\frac{\partial \Phi _{y}^{\prime }}{\partial z^{\prime }}\alpha  \\ 
\frac{\partial \Phi _{x}^{\prime }}{\partial t^{\prime }}\beta _{z}-i\beta_{y}\frac{\partial \Phi _{x}^{\prime }}{\partial x^{\prime }}+i\beta _{x} \frac{\partial \Phi _{x}^{\prime }}{\partial y^{\prime }}+\frac{\partial \Phi _{x}^{\prime }}{\partial z^{\prime }}\alpha -\frac{\partial \Phi_{z}^{\prime }}{\partial t^{\prime }}\beta _{x}-\frac{\partial \Phi_{z}^{\prime }}{\partial x^{\prime }}\alpha +i\beta _{z}\frac{\partial \Phi_{z}^{\prime }}{\partial y^{\prime }}-i\beta _{y}\frac{\partial \Phi_{z}^{\prime }}{\partial z^{\prime }} \\ 
\frac{\partial \Phi _{y}^{\prime }}{\partial t^{\prime }}\beta _{x}+\frac{\partial \Phi _{y}^{\prime }}{\partial x^{\prime }}\alpha -i\beta _{z}\frac{\partial \Phi _{y}^{\prime }}{\partial y^{\prime }}+i\beta _{y}\frac{\partial \Phi _{y}^{\prime }}{\partial z^{\prime }}-\frac{\partial \Phi_{x}^{\prime }}{\partial t^{\prime }}\beta _{y}-i\beta _{z}\frac{\partial\Phi _{x}^{\prime }}{\partial x^{\prime }}-\frac{\partial \Phi _{x}^{\prime }}{\partial y^{\prime }}\alpha +i\beta _{x}\frac{\partial \Phi _{x}^{\prime }}{\partial z^{\prime }}%
\end{bmatrix}%
=$

\begin{equation}\label{eq:2.1e}
=\alpha \frac{\partial \pmb{\Phi} ^{\prime }}{\partial t^{\prime }}+i\alpha \nabla ^{\prime }\times \pmb{\Phi }^{\prime } +i\pmb{\beta } \times \frac{\partial \pmb{\Phi}^{\prime }}{\partial t^{\prime }}+\pmb{\beta}(\nabla^{\prime }\pmb{\Phi }^{\prime })-\pmb{\beta } \times (\nabla ^{\prime }\times \pmb{\Phi}^{\prime })
\end{equation}

Substituting all the above partial results \eqref{eq:2.1b} - \eqref{eq:2.1e} into the equation \eqref{eq:2.1a} we obtain

\[
\begin{bmatrix}
\frac{\partial }{\partial t} \\ 
\nabla
\end{bmatrix}
\begin{bmatrix}
\varphi (X) \\ 
\pmb{\Phi }(X)
\end{bmatrix}=
\begin{bmatrix}
\alpha \left( \frac{\partial \varphi ^{\prime }}{\partial t^{\prime }}+\nabla^{\prime} \pmb{\Phi }^{\prime }\right) + \pmb{\beta } \left(\nabla ^{\prime }
\varphi ^{\prime }+\frac{\partial \Phi^{\prime }}{\partial t^{\prime }} +i\nabla^{\prime} \times 
\pmb{\Phi }^{\prime }\right) \\ 
\alpha \left(\nabla ^{\prime }
\varphi ^{\prime }+\frac{\partial \Phi^{\prime }}{\partial t^{\prime }} +i\nabla^{\prime} \times 
\pmb{\Phi }^{\prime }\right)

+\pmb{\beta} (\frac{\partial \varphi ^{\prime }}{\partial t^{\prime }}+\nabla^{\prime} \pmb{\Phi }^{\prime })

+i\pmb{\beta} \times \left(\frac{\partial \Phi^{\prime }}{\partial t^{\prime }} + \nabla ^{\prime }
\varphi ^{\prime } +i\nabla^{\prime} \times 
\pmb{\Phi }^{\prime }\right)
\end{bmatrix}
=
\begin{bmatrix}
\alpha \\ 
\pmb{\beta }%
\end{bmatrix}%
\begin{bmatrix}
\frac{\partial }{\partial t^{\prime}} \\ 
\nabla^{\prime}%
\end{bmatrix}%
\begin{bmatrix}
\varphi ^{\prime } \\ 
\pmb{\Phi }^{\prime }%
\end{bmatrix},
\]

which completes the proof of the 1st identity of Theorem \ref{th:2}.

2. Proof of the identity $\partial ^{-}A\left( X\right) =\partial ^{\prime -}\Gamma ^{-}A\left(X^{\prime }\Gamma ^{-1}\right)$ is left to the reader.

\end{myproof}

\newpage
\begin{theorem} \label{th:1.3}

Let $A(X)$ be an analytic paravector function defined on the set $C^{1+3}$, then for each paravector $\Gamma$ it is true that:

\begin{enumerate}
\item $\partial \left[ A\left(X\right) \Gamma \right] = \left[ \partial A\left( X\right) \right] \Gamma$
\item $\partial ^{-}\left[A\left( X\right) \Gamma \right] = \left[ \partial ^{-}A\left( X\right) \right] \Gamma$
\end{enumerate}
\end{theorem}

\begin{proof}1.

$\partial \left[ A\left(X\right) \Gamma \right] =
\begin{bmatrix}
\frac{\partial}{\partial t} \\ 
\nabla
\end{bmatrix}
(\begin{bmatrix}
\varphi (X) \\ 
\pmb{\Phi}(X)
\end{bmatrix}
\begin{bmatrix}
\alpha \\ 
\pmb{\beta}
\end{bmatrix})=
\begin{bmatrix}
\frac{\partial}{\partial t} \\ 
\nabla
\end{bmatrix}
\begin{bmatrix}
\alpha \varphi (X) + \pmb{\Phi}(X)\pmb{\beta}\\ 
\alpha\pmb{\Phi}(X)+\pmb{\beta}\varphi (X) +i\pmb{\Phi}(X)\times\pmb{\beta}
\end{bmatrix}=$

\noindent $=\begin{bmatrix}
\alpha \frac{\partial}{\partial t}\varphi(X)
+\pmb{\beta}\frac{\partial}{\partial t}\pmb{\Phi}(X)
+\alpha\nabla\pmb{\Phi}(X)
+\pmb{\beta}\nabla\varphi(X)
+i\nabla[\pmb{\Phi}(X)\times\pmb{\beta}] \\
\alpha\frac{\partial}{\partial t}\pmb{\Phi}(X)
+\pmb{\beta}\frac{\partial}{\partial t}\varphi (X)
+\alpha\nabla\varphi(X)
+i[\frac{\partial}{\partial t}\pmb{\Phi}(X)\times\pmb{\beta}
+\alpha\nabla\times\pmb{\Phi}(X)+
\nabla\varphi(X)\times\pmb\beta]
+\nabla[\pmb{\Phi}(X)\pmb{\beta}]
-\nabla\times[\pmb{\Phi}(X)\times\pmb{\beta}]
\end{bmatrix},
$

hence, under property of the nabla operator we obtain

$=\begin{bmatrix}
[\frac{\partial}{\partial t}\varphi(X)
+\nabla\pmb{\Phi}(X)]\alpha
+[\frac{\partial}{\partial t}\pmb{\Phi}(X)
+\nabla\varphi(X)
+i\nabla\times\pmb{\Phi}(X)]\pmb{\beta} \\
[\frac{\partial}{\partial t}\pmb{\Phi}(X)
+\nabla\varphi(X)
+i\nabla\times\pmb{\Phi}(X)]\alpha
+[\frac{\partial}{\partial t}\varphi(X)
+\nabla\pmb{\Phi}(X)]\pmb{\beta}
+i[\frac{\partial}{\partial t}\pmb{\Phi}(X)
+\nabla\varphi(X)
+i\nabla\times\pmb{\Phi}(X)]\times\pmb{\beta}
\end{bmatrix}=
$

$=\left[ \partial A\left(X\right) \right] \Gamma$

2. Similarly, as the proof of the equation 1.
\end{proof}

Formulas of transformation of the field by the rotation of reference system $X^{\prime} = \Gamma X \Gamma^{-1}$ follow from above results, where the rotation means a more general transformation then Euclidean rotation \cite{Radomanski}.

\begin{example}Rotation of the observer in the field \label{ex2}

Let $\Lambda$ be an orthogonal paravector (ie. det $\lambda=1$) and let the fields $A(X)$ and $B(X)$ satisfy the relationship $\partial A(X)=B(X)$, where $X \in C^{1+3}$. The observer rotates:
\[\partial A(\Lambda ^{-}X^{\prime}\Lambda) = B(\Lambda ^{-}X^{\prime}\Lambda)\qquad, \qquad\qquad \text{where}\qquad X^{\prime} = \Lambda X\Lambda ^{-}\]

In the turned frame the above equation has the form (by Theorems \ref{th:1} and \ref{th:2})
\[\Lambda^{-}\partial^{\prime} \Lambda A(\Lambda ^{-}X^{\prime}\Lambda) = B(\Lambda ^{-}X^{\prime}\Lambda)\]

Multiplying this equation on the left-side by $\Lambda$ and right-side by $\Lambda^{-}$, on the basis of the Theorem \ref{th:1.3}, we obtain an equation of the field after rotation.
\[\partial^{\prime} [\Lambda A(\Lambda ^{-}X^{\prime}\Lambda)\Lambda^{-}] = \Lambda[B(\Lambda ^{-}X^{\prime}\Lambda)]\Lambda^{-},\]

Similarly for the reversed operator (4-gradient). The conclusion is obvious:
\textbf{If the observer turns to one side, the field around it will turn by the same amount in the opposite direction}.
\end{example}

\section{Invariance of wave equation under orthogonal transformation.}

Using the theorems \ref{th:1} - \ref{th:1.3}, we can easily demonstrate the invariance of the wave equation $\square A(X)= B(X)$ under the transformation represented by the orthogonal paravector. We can do this in four ways:

\begin{enumerate}
\item $\square A(X)=\partial ^{-}\partial A(X)=\partial ^{\prime -}\Lambda
^{-}\Lambda \partial ^{\prime }A(X^{\prime }\Lambda ^{-})=\square ^{\prime
}A(X^{\prime }\Lambda ^{-})=B(X^{\prime }\Lambda ^{-}), \qquad$ or
\begin{equation}\label{eq:3.1}
\square A(X)=B(X)\qquad \iff \qquad \square ^{\prime
}A(X^{\prime }\Lambda ^{-})=B(X^{\prime }\Lambda ^{-})
\end{equation}
\item $\square A(X)=\partial \partial ^{-}A(X)=\Lambda \partial ^{\prime}\partial ^{\prime -}[\Lambda ^{-}A(X^{\prime }\Lambda ^{-})]=\Lambda
\square^{\prime }[\Lambda ^{-}A(X^{\prime }\Lambda ^{-})]=B(X^{\prime
}\Lambda ^{-}), \qquad $ hence
\begin{equation}\label{eq:3.2}
\square A(X)=B(X)\qquad \iff \qquad \square^{\prime }[\Lambda
^{-}A(X^{\prime }\Lambda ^{-})]=\Lambda ^{-}B(X^{\prime }\Lambda ^{-})
\end{equation}
\item $\square A(X)=\partial ^{-}\partial A(X)=\Lambda ^{-}\partial ^{\prime -}\partial ^{\prime } \left[ \Lambda A\left( \Lambda ^{-}X^{\prime }\right) \right]=\Lambda
^{-}\square^{\prime }[\Lambda A(\Lambda ^{-}X^{\prime })]=B(\Lambda
^{-}X^{\prime }),\qquad $ hence
\begin{equation}\label{eq:3.3}
\square A(X)=B(X)\qquad \iff \qquad \square^{\prime
}[\Lambda A(\Lambda ^{-}X^{\prime })]=\Lambda B(\Lambda ^{-}X^{\prime })
\end{equation}
\item $\square A(X)=\partial \partial ^{-}A(X)=\partial ^{\prime }\Lambda\Lambda ^{-}\partial ^{\prime -}A(\Lambda ^{-}X^{\prime })=\square ^{\prime
}A(\Lambda ^{-}X^{\prime })=B(\Lambda ^{-}X^{\prime }),\qquad \qquad $ or
\begin{equation}\label{eq:3.4}
\qquad \square A(X)=B(X)\qquad \iff \qquad \square ^{\prime
}A(\Lambda ^{-}X^{\prime })=B(\Lambda ^{-}X^{\prime })
\end{equation}
\end{enumerate}

From the above relationships it can be seen that further discussion can be carried out in different directions. In points 1 and 4 both the equation and the value of the function are invariant.
\[
A^{\prime }=A\qquad \text{i} \qquad B^{\prime }=B
\]

At points 2 and 3 the form of wave equation is invariant, while the values of the function of a field undergo changes:
\begin{itemize}
\item contravariant
\[
A^{\prime}=\Lambda^{-} A \qquad \text{i} \qquad B^{\prime }=\Lambda^{-} B
\]
\item or covariant
\[
A^{\prime }= \Lambda A \qquad \text{i} \qquad B^{\prime }=\Lambda B
\]
\end{itemize}

We encounter an interesting problem that goes beyond the established frame to show in this article, and enters the field of physics, so its analysis will be presented in another paper.

As for the rotation of the observer in the field meeting the wave equation, we get the same result as in example \ref{ex2}.
\[\square A(X)=B(X)\qquad \iff \qquad \square ^{\prime
} \Lambda [A(\Lambda ^{-}X^{\prime }\Lambda)]\Lambda ^{-}=
\Lambda [B(\Lambda ^{-}X^{\prime }\Lambda)]\Lambda ^{-} \]

\section*{Summary}

The wave equation is one of the most important relationships in physics. It underlies the theory of the electromagnetic field and relativistic quantum mechanics, and is applicable in all fields of physics. The considerations presented above and simplicity of calculation show that the paravector calculus fits into this equation naturally and so, it is also natural for relativistic physics. I will by trying to convince the reader in further publications that it is so and that using paravectors we can take a different look at the seemingly well-known phenomenons.    

\bibliographystyle{plain}

\end{document}